%% file: main-paper.tex
\documentclass[sigconf]{acmart}
\allowdisplaybreaks

\usepackage{amssymb,amsthm}

\theoremstyle{plain}
\newtheorem{theorem}{Theorem}

\newtheorem{prop}[theorem]{Proposition}

\newtheorem{assumption}{Assumption}

\theoremstyle{definition}

\theoremstyle{remark}
\newtheorem*{remark}{Remark}

\DeclareMathOperator*{\argmax}{arg\,max}

\newcommand{\E}{\mathbb{E}}

\renewcommand{\epsilon}{\varepsilon}
\renewcommand{\phi}{\varphi}
\def\BibTeX{{\rm B\kern-.05em{\sc i\kern-.025em b}\kern-.08em
    T\kern-.1667em\lower.7ex\hbox{E}\kern-.125emX}}

\usepackage{algorithm}
\usepackage{algorithmic}
\usepackage{enumitem}

\input{rights}

\begin{document}

\title{Windex: Index-Driven Contextual Scheduler for Scalable Service Guarantees in NextG Cellular Networks}

\author{Archana Bura}
\email{abura@ucsd.edu}
\affiliation{\institution{University of California San Diego}\city{La Jolla}\state{California}\country{USA}}
\author{Ushasi Ghosh}
\email{ughosh@ucsd.edu}
\affiliation{\institution{University of California San Diego}\city{La Jolla}\state{California}\country{USA}}
\author{Alexander Riedel}
\email{a.r.tamu@tamu.edu}
\affiliation{\institution{Texas A\&M University}\city{College Station}\state{Texas}\country{USA}}
\author{Dinesh Bharadia}
\email{dineshb@ucsd.edu}
\affiliation{\institution{University of California San Diego}\city{La Jolla}\state{California}\country{USA}}
\author{Srinivas Shakkottai}
\email{sshakkot@tamu.edu}
\affiliation{\institution{Texas A\&M University}\city{College Station}\state{Texas}\country{USA}}

\renewcommand{\shortauthors}{Bura et al.}

\begin{abstract}
Next-generation cellular networks must support heterogeneous service objectives, ranging from throughput-oriented eMBB to latency- and regularity-constrained URLLC and XR traffic. Enforcing these objectives at the RAN scheduler is difficult because wireless channels, queues, and application arrivals evolve at the millisecond timescale while spectrum resources are shared across users. We formulate service-aware RAN scheduling as a constrained Markov decision process with per-user throughput and service-regularity constraints. A Lagrangian relaxation gives an exact decomposition of the relaxed problem into per-user control problems with a common resource price. We prove fixed-dual indexability for an illustrative time-since-last-service (TSLS) debt arm with a throughput objective and TSLS constraint, note the analogous scalar throughput-debt result, and assume fixed-dual indexability for the full vector state used in deployment. We design \emph{Windex}, an index-driven scheduler that learns dual-dependent per-user break-even resource prices using compact neural approximators over queue, channel, TSLS, and violation-deficit state. We further design \emph{CWindex}, a contextual extension that infers service class from queue-history time series when explicit labels are unavailable. 
We implement Windex as an application compatible with a realtime 5G radio intelligent controller.
Inference takes 10--20~$\mu$s per user on a general-purpose CPU, well within a 1~ms control budget. Trace-driven emulation using over-the-air channel measurements and larger simulations show that Windex and CWindex substantially reduce per-user service violations while maintaining throughput efficiency.
\end{abstract}

\input{ccs}
\maketitle

\input{01-intro_v0}
\input{02-background_v0}

\input{03-problem_v1}
\input{04-solution_v3}
\input{05-design_v2}

\input{06-system}

\input{07-evaluations}
\input{08-discussions}

\begin{acks}
This research was supported in part by NSF grants CNS 2312978 and CNS 2312979. All opinions are those of the authors and do not necessarily reflect those of the sponsoring agency.
\end{acks}

\bibliographystyle{ACM-Reference-Format}
\bibliography{reference}

\end{document}

%% file: rights.tex
\copyrightyear{2026}
\acmYear{2026}
\setcopyright{cc}
\setcctype{by}
\acmConference[MobiHoc '26]{The Twenty-seventh International Symposium on Theory, Algorithmic Foundations, and Protocol Design for Mobile Networks and Mobile Computing}{November 23--26, 2026}{Tokyo, Japan}
\acmBooktitle{The Twenty-seventh International Symposium on Theory, Algorithmic Foundations, and Protocol Design for Mobile Networks and Mobile Computing (MobiHoc '26), November 23--26, 2026, Tokyo, Japan}
\acmDOI{10.1145/3842721.3849659}
\acmISBN{979-8-4007-2975-1/2026/11}

%% file: ccs.tex
\ccsdesc[500]{Networks~Wireless access networks}
\ccsdesc[500]{Networks~\texorpdfstring{\newline}{ }Network resources allocation}
\keywords{RAN scheduling, Whittle indices, TSLS, resource pooling}

%% file: 01-intro_v0.tex
\section{Introduction}
NextG cellular networks must support applications with very different service objectives. Enhanced mobile broadband (eMBB) flows are primarily throughput driven, ultra-reliable low-latency communication (URLLC) and massive machine-type communication (mMTC) flows require regular service at low or moderate rates, and extended reality (XR) combines high throughput with tight service regularity. These requirements must be met under fast channel variation caused by mobility, blockage, and interference. A scheduler that only optimizes aggregate throughput or fairness can therefore satisfy the network average while still violating the service objective of an individual user.

Providing per-user guarantees is difficult because scheduling is a sequential and coupled control problem. The RAN must track each user's queue, channel, recent service history, and accumulated service deficit, and then select a feasible allocation from a combinatorial set of resource-block assignments every transmission time interval (TTI). Slicing-based systems reduce this complexity by assigning resources to coarse service classes or slices, but slice-level averages do not directly control per-user violations when the offered load approaches the available capacity~\cite{kokku2011nvs,radiosaber}. A near-congested slice may therefore contain both over-served and under-served users even when the slice-level target is met.

This paper argues for a finer-grained scheduler built around an index structure. Many queueing control problems have the property that a user's urgency can be summarized by a scalar priority: long queues, large time-since-last-service (TSLS), poor accumulated service, or a favorable channel can all increase the marginal value of serving that user. Whittle index policies exploit this structure by decoupling a multi-user restless-bandit problem into per-user control problems and ranking users by a state-dependent index~\cite{whittle}. This ranking structure is attractive for RAN scheduling because it turns a coupled combinatorial problem into per-user inference followed by a top-$R$ selection.

We present \emph{Windex}, a service-aware MAC scheduler that uses learned Whittle-style indices for real-time RAN control. Windex formulates per-user throughput and service-regularity requirements as constraints, relaxes the shared resource budget with a resource price, and learns the per-user break-even price at which high service becomes worthwhile. The deployed scheduler uses the learned index as a priority weight: in each TTI it computes one score per user, selects the top-ranked users subject to the high-resource budget, and maps the binary high/low decision to the resource-block allocations supported by the MAC scheduler.

\begin{figure}[!htbp]
    \centering
    \includegraphics[width=\columnwidth]{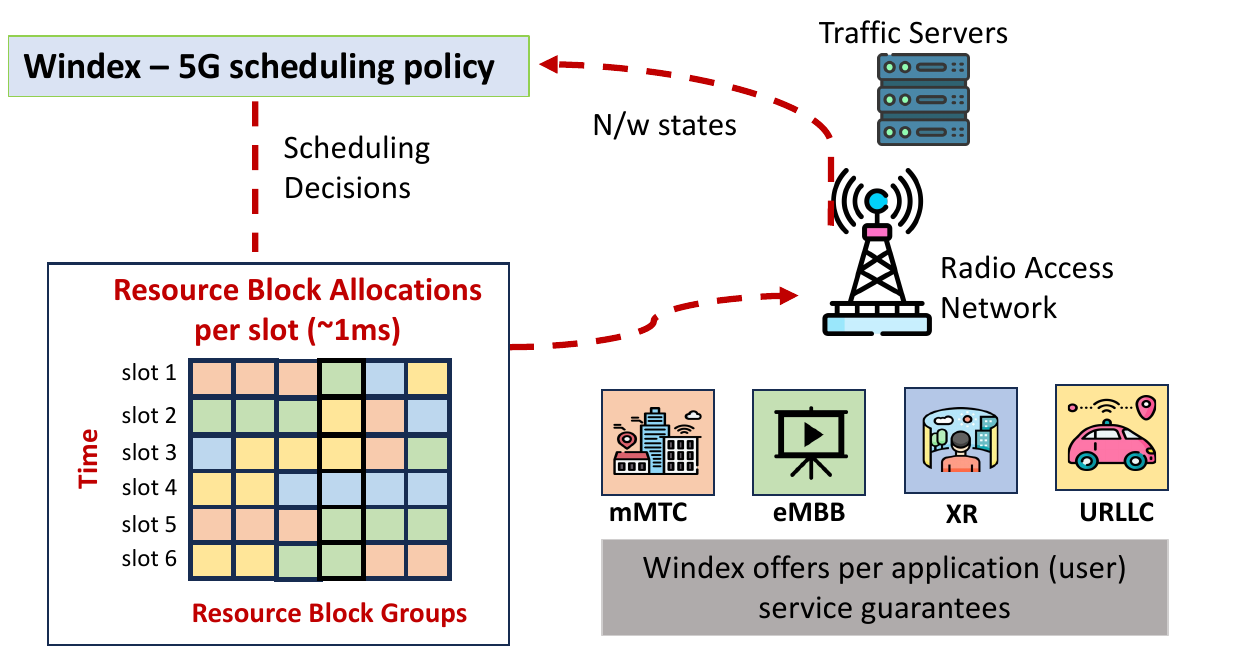}
    \caption{Windex provides MAC-layer scheduling decisions for heterogeneous service classes by ranking users according to learned service-aware indices.}
    \Description{A shared radio resource pool serves users from heterogeneous service classes. Windex computes a learned scalar priority for each user and passes the ranking to the MAC scheduler.}
    \label{fig:WindexMainfigure}
\end{figure}

Our analysis establishes exact fixed-dual decomposition and indexability for an illustrative scalar TSLS-debt model. This motivates the deployed contextual scheduler. We state the indexability assumption for its richer vector state separately.

We also introduce \emph{CWindex}, a context-inference layer for deployments where the service class of a flow is unavailable, noisy, or changes over time. CWindex uses a structured state-space sequence model to classify flows from recent backlog dynamics, then selects the corresponding Windex model. CWindex preserves the real-time interface while inferring service context from traffic.

Our contributions are:
\begin{enumerate}[leftmargin=*]
    \item We formulate per-user service-aware RAN scheduling as a constrained MDP with per-user throughput and TSLS constraints and a high/low RBG allocation interface.
    \item We derive the Lagrangian decomposition of the relaxed problem, prove indexability for an illustrative TSLS-debt arm, and state the fixed-dual indexability assumption used for the full vector per-user state learned by Windex.
    \item We design Windex, a compact learned-index scheduler trained from simulated episodes with service-class rewards and adapted online through measured violation feedback. Index computation takes 10--20~$\mu$s per UE and grows linearly with UE count.
    \item We design CWindex, an S4-based (structured state-space model) context network that infers eMBB, URLLC, and XR service classes from queue-history time series and integrates with Windex without changing the MAC scheduler interface.
    \item We implement Windex in EdgeRIC~\cite{edgeric} and evaluate it with over-the-air channel traces, trace-driven emulation, and larger simulations across heterogeneous service mixes. The simulations include up to 60 UEs and a separate comparison of slice-level allocation with resource pooling. Measured index computation for 20 UEs fits within the 1~ms control budget.
\end{enumerate}

%% file: 02-background_v0.tex
\section{Background and Related Work}

\textbf{Scheduling and slicing in 5G.}
Network slicing allocates resources among service classes, with a slice controller setting shares and a MAC scheduler serving users within each slice every TTI. Open RAN implementations such as NexRAN and other O-RAN testbeds follow this architecture~\cite{nexran,nextGintelligence}. Channel-aware slicing adapts shares to load and channel state~\cite{radiosaber,kokku2011nvs}, with guarantees generally expressed as long-term slice averages. Zipper~\cite{balasingam2024application} targets application-level throughput and latency assurance through predictive slice control and per-TTI MAC scheduling. Windex learns priorities for individual users, allowing throughput and service regularity to guide allocation across a shared resource pool.

Classical MAC schedulers offer fast priority rules. Max-CQI exploits favorable channels, proportional fairness balances rate and fairness~\cite{huang2018gpf}, and max-weight uses backlog and channel quality.

\textbf{Service regularity and information freshness.}
A throughput target quantifies service volume, while time since last service (TSLS) captures how that service is spaced over time. The scheduler can update TSLS directly from service events at each TTI. A scheduler can meet an average rate by concentrating transmissions into bursts, even as long gaps between them affect applications that require regular service. Prior work combines TSLS with queue length to obtain throughput optimality and service-regularity guarantees~\cite{li2013regulatedservice,li2015regularservice}. Age of information (AoI) measures freshness from an update's generation time and thus includes queueing delay. Threshold policies and indexability for a relaxed AoI problem with throughput constraints and on-demand updates have been established~\cite{kadota2019throughput}.

\textbf{Whittle indices in networks.}
Whittle's relaxation separates a coupled allocation problem into independent per-arm control problems through a common resource price~\cite{whittle}. This price coordinates the separate per-user decisions through the available resource budget. The resulting indices reduce joint action selection to a ranking of individual users. Whittle index policies have been applied to downlink scheduling with imperfect channel state~\cite{ouyang}, AoI scheduling under throughput constraints~\cite{kadota2019throughput}, and opportunistic wireless scheduling~\cite{borkaroppurtunistic}. Their indexability arguments rely on problem-specific assumptions. Windex uses this decomposition to learn each user's priority from queue, channel, TSLS, and service-deficit state. This scalar interface lets heterogeneous service objectives share a common rule for ranking users and allocating resources.

\textbf{Learning-based resource allocation.}
Reinforcement learning has been used for adaptive video streaming~\cite{mao2017neural,bhattacharyya2019qflow}, LTE downlink scheduling~\cite{sharma2020deep}, and cellular traffic management~\cite{chinchali2018cellular}. Deployment in a 1~ms RAN loop requires action selection that scales with the number and mixture of users. Directly choosing a scheduled subset gives a joint action space that grows combinatorially with the user population. DeDe~\cite{xu2025dede} demonstrates scalable resource allocation through decomposition. Windex follows NeurWIN~\cite{neurwin} by learning per-user indices, with independent inference followed by top-$R$ selection. Each network evaluates one user's state, allowing the same model to serve different population sizes.

\textbf{Context inference.}
Structured state-space models such as S4~\cite{gu2021efficiently} can classify long time series with low inference cost. CWindex supplies missing service context from backlog history.

%% file: 03-problem_v1.tex
\section{Problem Formulation}
\label{sec:services}

We consider a RAN serving $N$ downlink users. Each user $n$ belongs to a service class $k(n)\in\{\mathrm{eMBB},\mathrm{URLLC},\mathrm{mMTC},\mathrm{XR}\}$ and has two service objectives: a minimum average throughput target $B_n$ and a maximum average service-regularity target $L_n$. We measure regularity by the time since last service (TSLS), denoted $\tau_n(t)$, or by a thresholded TSLS violation indicator. The scheduler runs every TTI and observes a per-user state
\begin{equation}
    s_n(t) = \big(q_n(t), c_n(t), \tau_n(t), d^B_n(t), d^L_n(t), k(n)\big),
\end{equation}
where $q_n(t)$ is downlink queue backlog, $c_n(t)$ is channel quality, $d^B_n(t)$ is a throughput-deficit statistic, and $d^L_n(t)$ is a TSLS-violation statistic. The deficit terms are empirical summaries of recent constraint violations. The dual formulation in Section~\ref{sec:solution} motivates their use as service-pressure signals, while the implementation feeds measured violation statistics directly to the learned priority network. CWindex infers missing class labels from queue history.

\begin{figure}[t]
\centering
\includegraphics[width=\columnwidth]{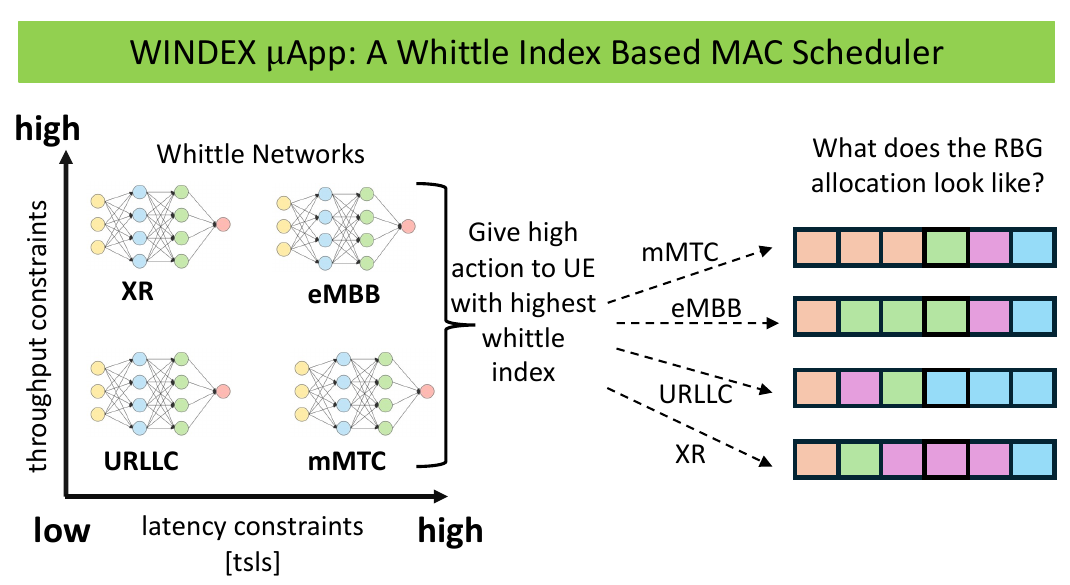}
\caption{Windex ranks users using queue, channel, TSLS, and violation feedback, then allocates RBGs by priority.}
\Description{Each user's queue, channel quality, time since last service, and throughput and regularity feedback enter its service-class network. The output priorities determine the shared high-resource allocation.}
\label{fig:WindexNetworks}
\end{figure}

\subsection{High/low resource abstraction}
The MAC scheduler allocates resource-block groups (RBGs), but the learning and control problem uses a binary abstraction. In each TTI, user $n$ receives either a high action $a_n(t)=1$ or a low action $a_n(t)=0$. The high action corresponds to a larger RBG allocation and the low action corresponds to a smaller allocation or no allocation, depending on the scenario. In the trace-driven emulator, for example, the high/low/no-service mapping is implemented as 9, 2, or 0 RBGs per TTI. The MAC assigns the high allocation to the highest-ranked UEs, the low allocation to the next configured number, and zero to the remainder. The configured counts enforce the total RBG budget.

Let $R$ denote the number of users that can receive the high action in a TTI. The physical resource constraint is
\begin{equation}
    \sum_{n=1}^{N} a_n(t) \le R, \qquad \forall t.
    \label{eq:hard_resource_constraint}
\end{equation}
The achieved throughput of user $n$ is $r^B_n(s_n(t),a_n(t))$, and its regularity cost is $r^L_n(s_n(t),a_n(t))$. We use $B$ for the throughput objective and $L$ for the latency/regularity objective to avoid sign ambiguity: throughput should be large, while regularity cost should be small.

\subsection{Constrained scheduling objective}
With discount factor $\gamma\in(0,1)$, the constrained objective is
\begin{align}
\max_{\pi}\quad &
    \E_{\pi}\!\left[\sum_{t=0}^{\infty}\gamma^t
    \sum_{n=1}^{N} r^B_n(s_n(t),a_n(t))\right]
    \label{eq:cmdp_obj}\\
\text{s.t.}\quad &
    \sum_{n=1}^{N}a_n(t)\le R, \qquad \forall t,
    \label{eq:cmdp_resource}\\
&   \E_{\pi}\!\left[\sum_{t=0}^{\infty}\gamma^t r^B_n(s_n(t),a_n(t))\right]
    \ge \frac{B_n}{1-\gamma}, \qquad \forall n,
    \label{eq:cmdp_tpt}\\
&   \E_{\pi}\!\left[\sum_{t=0}^{\infty}\gamma^t r^L_n(s_n(t),a_n(t))\right]
    \le \frac{L_n}{1-\gamma}, \qquad \forall n.
    \label{eq:cmdp_tsls}
\end{align}
The normalization by $(1-\gamma)^{-1}$ makes $B_n$ and $L_n$ comparable to per-slot targets. The expectation is over arrivals, channel evolution, and any policy randomization. The standard restless-bandit abstraction used here assumes that, conditional on local state and local action, each user's queue and channel state evolve independently of other users. The users are coupled by the shared high-action constraint in~\eqref{eq:hard_resource_constraint}. If the offered load makes the constraints infeasible, the goal is to minimize the resulting violations while maintaining throughput efficiency. We report the resulting violation magnitudes.

The exact problem is hard for two reasons. First, the action vector $a(t)\in\{0,1\}^{N}$ must satisfy a combinatorial resource constraint. Second, the service constraints are per-user and long-term: the effect of a scheduling decision depends on the user's future queue, channel, and TSLS evolution. Windex addresses the first difficulty by using a Lagrangian decomposition and the second by learning a per-user index over the service state.

%% file: 04-solution_v3.tex
\section{Solution Methodology: From CMDP to Learned Indices}
\label{sec:solution}

Windex uses Whittle's relaxation to convert the coupled scheduling CMDP in Section~\ref{sec:services} into per-user control problems. We first establish exact decomposition for fixed dual variables, then prove indexability for an illustrative scalar TSLS-debt arm. We identify the conditions needed for a throughput-debt analogue and state the fixed-dual indexability assumption for the full vector state. This structure motivates the violation feedback used by the deployed scheduler.

The analysis is stated for the high/low scheduling abstraction used throughout the paper.  In the implementation, the high action corresponds to assigning the larger RBG allocation to a UE and the low action corresponds to assigning the smaller allocation, possibly zero.  The index policy selects the high-action users, and the MAC layer maps each action to its configured RBG allocation.

\subsection{Lagrangian Relaxation and Decomposition}
\label{subsec:lagrangian_decomposition}

The hard scheduling constraint in~\eqref{eq:cmdp_resource} couples users because at most $R$ users can receive the high action in a slot.  Whittle's relaxation replaces this hard per-slot constraint by a discounted average resource constraint and introduces a resource price $\lambda\ge0$.  We also introduce nonnegative per-user dual variables $\mu^B_n$ and $\mu^L_n$ for the throughput and TSLS constraints, respectively.  For fixed dual variables, the per-user Lagrangian stage reward is
\begin{equation}
    g_n(s,a;\lambda,\mu_n)
    = (1+\mu^B_n) r^B_n(s,a)
      - \mu^L_n r^L_n(s,a)
      - \lambda a,
    \label{eq:stage_lagrangian}
\end{equation}
where $\mu_n=(\mu^B_n,\mu^L_n)$.  Here $r^B_n$ denotes the throughput reward and $r^L_n$ denotes the TSLS or regularity cost.  Thus, increasing $\mu^B_n$ raises the value of serving users whose throughput requirement is difficult to meet, while increasing $\mu^L_n$ penalizes states with poor service regularity.  The target terms contribute constants to the dual objective. At fixed $(\lambda,\mu_n)$, the stage reward determines the maximizing policy.

For each user, define the fixed-dual value function
\begin{equation}
V_n(s;\lambda,\mu_n)=
\max_{\pi_n}\E_{\pi_n}\left[
\sum_{t=0}^{\infty}\gamma^t
 g_n(s_n(t),a_n(t);\lambda,\mu_n)
\,\middle|\,s_n(0)=s\right].
\label{eq:per_user_value}
\end{equation}
It satisfies the Bellman equation
\begin{align}
V_n(s;\lambda,\mu_n)=
\max_{a\in\{0,1\}}\Big\{
    g_n(s,a;\lambda,\mu_n)
    +\gamma\E[V_n(S'_n;\lambda,\mu_n)\mid s,a]
\Big\}.
\label{eq:per_user_bellman}
\end{align}

\begin{prop}[Exact decomposition for fixed dual variables]
\label{prop:decomposition}
Fix $\lambda\ge0$ and $\mu_n\ge0$ for all users.  The Lagrangian relaxation of~\eqref{eq:cmdp_obj}--\eqref{eq:cmdp_tsls} decomposes into $N$ independent per-user MDPs of the form~\eqref{eq:per_user_value}.  The associated dual objective can be written as
\begin{align}
D(\lambda,\mu)=
&\frac{\lambda R}{1-\gamma}
-\sum_{n=1}^{N}\frac{\mu^B_n B_n}{1-\gamma}
+\sum_{n=1}^{N}\frac{\mu^L_n L_n}{1-\gamma}
+\sum_{n=1}^{N}V_n(s_n(0);\lambda,\mu_n).
\label{eq:dual_objective}
\end{align}
\end{prop}

\begin{proof}
Add the relaxed constraints to~\eqref{eq:cmdp_obj} with nonnegative multipliers. The resource term is
\[
    \sum_{t\ge0}\gamma^t \lambda\left(R-\sum_{n=1}^N a_n(t)\right)
    = \frac{\lambda R}{1-\gamma}
      -\sum_{n=1}^N\sum_{t\ge0}\gamma^t\lambda a_n(t).
\]
For user $n$, the throughput lower bound contributes
\[
    \mu^B_n\left(\sum_{t\ge0}\gamma^t r^B_n(s_n(t),a_n(t))-\frac{B_n}{1-\gamma}\right).
\]
For user $n$, the TSLS upper bound contributes
\[
    \mu^L_n\left(\frac{L_n}{1-\gamma}-\sum_{t\ge0}\gamma^t r^L_n(s_n(t),a_n(t))\right).
\]
Terms involving $R$, $B_n$, and $L_n$ are constant for fixed duals and yield the first three terms of~\eqref{eq:dual_objective}. The policy-dependent remainder is
\[
    \sum_{n=1}^N \sum_{t\ge0}\gamma^t
    \Big((1+\mu^B_n)r^B_n(s_n(t),a_n(t))
        -\mu^L_n r^L_n(s_n(t),a_n(t))
        -\lambda a_n(t)\Big),
\]
which is exactly the sum of the per-user stage rewards in~\eqref{eq:stage_lagrangian}.  No term in this relaxed objective depends jointly on the actions of two different users.  Under the standard restless-bandit assumption that a user's next state depends only on its own state, action, traffic process, and channel process, the maximization over joint policies therefore separates into $N$ single-user maximizations.  Each separated maximization is the MDP in~\eqref{eq:per_user_value}, and substituting their optimal values gives~\eqref{eq:dual_objective}.
\end{proof}

Proposition~\ref{prop:decomposition} motivates one priority per user. Whittle indices supply this priority for indexable per-user problems.

\subsection{Fixed-Dual Whittle Indices}
\label{subsec:fixed_dual_indices}

For fixed $\mu_n$, define the active and passive action values for user $n$
\begin{align}
Q^1_n(s;\lambda,\mu_n)
&=g_n(s,1;\lambda,\mu_n)
  +\gamma\E[V_n(S'_n;\lambda,\mu_n)\mid s,1],\label{eq:Q_active}\\
Q^0_n(s;\lambda,\mu_n)
&=g_n(s,0;\lambda,\mu_n)
  +\gamma\E[V_n(S'_n;\lambda,\mu_n)\mid s,0].\label{eq:Q_passive}
\end{align}
The passive set $\mathcal{P}_n(\lambda;\mu_n)=\{s:Q^0_n(s;\lambda,\mu_n)\ge Q^1_n(s;\lambda,\mu_n)\}$ contains states where leaving the user passive is optimal. The problem is \emph{indexable} if this set expands with price: $\mathcal{P}_n(\lambda_1;\mu_n)\subseteq\mathcal{P}_n(\lambda_2;\mu_n)$ whenever $\lambda_1\le\lambda_2$.
When this monotonicity holds, the Whittle index is the break-even resource price
\begin{equation}
    W_n(s;\mu_n)=
    \inf\{\lambda: s\in\mathcal{P}_n(\lambda;\mu_n)\}.
    \label{eq:whittle_break_even}
\end{equation}
A larger index means that a larger resource price is required before it becomes optimal to leave the user passive, and hence that the user has higher marginal value of service.

\subsection{Illustrative Indexability Result: TSLS Debt}
\label{subsec:deficit_model}

The full Windex state is vector-valued: backlog, CQI, TSLS, and violation deficits have different dynamics. To make the index structure explicit, we first isolate the TSLS component and prove indexability for a scalar regularity-debt arm. This special case corresponds to a throughput objective with a TSLS constraint: service earns a throughput reward, and TSLS incurs a dual-weighted cost.

Let $x_t\in\{0,1,\ldots,X\}$ denote regularity debt. Passivity increments it up to saturation, while high service resets it:
\begin{align}
    F_0(x)&=\min\{X,x+1\},\label{eq:tsls_passive_dyn}\\
    F_1(x)&=0.\label{eq:tsls_active_dyn}
\end{align}
A high action gives normalized throughput benefit $\rho\ge0$. Let $h_\mu(x)=\mu^L\ell(x)$ be the dual-weighted TSLS cost, where $\ell(0)=0$ and $\ell$ is nondecreasing. For fixed TSLS multiplier $\mu^L$ and resource price $\lambda$, the cost form of the single-user problem is
\begin{equation}
J_{\lambda,\mu}(x)=
\min_\pi \E_\pi\!\left[
\sum_{t=0}^{\infty}\gamma^t\big(h_\mu(x_t)+(\lambda-\rho)a_t\big)
\,\middle|\,x_0=x\right].
\label{eq:tsls_value}
\end{equation}
The term $-\rho a_t$ is the throughput reward written as a reduction in cost, while $\lambda a_t$ is the resource charge.

\begin{theorem}[Indexability of the TSLS-debt arm]
\label{thm:tsls_indexability}
For every fixed $\mu^L$, the TSLS-debt model in~\eqref{eq:tsls_value} is Whittle-indexable with respect to the resource price $\lambda$. An optimal policy is a threshold policy: for each $(\lambda,\mu^L)$ there is an $m(\lambda,\mu^L)\in\{0,1,\ldots,X+1\}$ such that the high action is optimal for $x\ge m(\lambda,\mu^L)$ and the passive action is optimal for $x<m(\lambda,\mu^L)$. The break-even index $W_{\mathrm{TSLS}}(x;\mu^L)$ defined by~\eqref{eq:whittle_break_even} is nondecreasing in $x$.
\end{theorem}

\begin{proof}
The Bellman equation is
\begin{equation}
J_{\lambda,\mu}(x)=h_\mu(x)+\min\{\gamma J_{\lambda,\mu}(F_0(x)),\;\lambda-\rho+\gamma J_{\lambda,\mu}(0)\}.
\label{eq:tsls_bellman}
\end{equation}
Starting from $J^0=0$, value iteration preserves monotonicity: $h_\mu$ and $F_0$ are nondecreasing, and the active continuation cost is constant in $x$. The minimum of the two nondecreasing action costs is nondecreasing. The discounted limit therefore remains nondecreasing.

The high action is optimal exactly when
\begin{equation}
    \lambda-\rho
    \le
    \gamma\big(J_{\lambda,\mu}(F_0(x))-J_{\lambda,\mu}(0)\big).
    \label{eq:tsls_active_condition}
\end{equation}
The right-hand side is nondecreasing in $x$, so choosing the passive action on ties gives an optimal threshold policy. Threshold $m=0$ gives continuous service, while $m=X+1$ leaves the arm passive.

For monotonicity in price, express the threshold-$m$ cost as
\begin{equation}
    C_m(x;\lambda,\mu)=H_m(x;\mu)+(\lambda-\rho)A_m(x),
    \label{eq:threshold_cost_compact}
\end{equation}
where $H_m$ is the discounted TSLS cost and $A_m$ is the discounted activation count. For $m\le X$, the first activation occurs after $(m-x)^+$ slots and subsequent activations are $m+1$ slots apart. Hence
\[
A_m(x)=\frac{\gamma^{(m-x)^+}}{1-\gamma^{m+1}},
\qquad A_{X+1}(x)=0.
\]
This count strictly decreases with $m$. Consequently, for $m_1<m_2$, the difference $C_{m_1}-C_{m_2}$ increases with $\lambda$: a more aggressive threshold cannot become newly preferable as service becomes more expensive. Choosing the largest optimal threshold, consistently with passivity on ties, therefore makes $m$ nondecreasing in $\lambda$. Its passive set $\{0,\ldots,m-1\}$ expands with price, proving indexability. This lower-interval form also implies that any price making $y>x$ passive makes $x$ passive, so $W_{\mathrm{TSLS}}(x;\mu^L)\le W_{\mathrm{TSLS}}(y;\mu^L)$.
\end{proof}

\begin{remark}[Throughput-debt extension]
Extending the argument to a finite discounted throughput-debt arm requires threshold optimality and ordered activation counts. The TSLS-debt proof suggests a straightforward procedure for checking these properties for a specified debt model. First, use value iteration and the Bellman action comparison to establish threshold optimality at each price, choosing passivity on ties. Next, verify that the discounted activation counts $A_m(x)$ are nonincreasing in threshold $m$ for every initial debt $x$. At fixed service multipliers and with price-independent transitions, costs of the form $c_\mu(x,a)+\lambda a$ give each threshold a discounted cost $C_m^0(x)+\lambda A_m(x)$, where $C_m^0$ excludes the resource charge. These ordered slopes make the largest Bellman-optimal threshold nondecreasing in price. Finally, verify that sufficiently low prices make activation uniquely optimal in every state and sufficiently high prices make passivity uniquely optimal. Together, these properties make the passive sets expand from empty to full, establishing indexability for the chosen model.
\end{remark}

\paragraph{Assumption for the full Windex state.}
The full analytical model uses the richer state $s=(q,c,\tau,d^B,d^L)$, where queue backlog, channel quality, TSLS, and violation deficits evolve jointly. Theorem~\ref{thm:tsls_indexability} establishes indexability for the specified scalar TSLS-debt dynamics. For the full vector state, we use the following fixed-dual assumption.

\begin{assumption}[Fixed-dual indexability of the full per-user problem]
\label{assump:fixed_dual_indexability}
For each service class $k$ and each fixed dual vector $\mu=(\mu^B,\mu^L)$, the per-user relaxed MDP over the vector state $s=(q,c,\tau,d^B,d^L)$ is indexable with respect to the resource price $\lambda$. Its passive set $\mathcal P_k(\lambda;\mu)$ expands monotonically as $\lambda$ increases.
\end{assumption}

Under Assumption~\ref{assump:fixed_dual_indexability}, the break-even price in~\eqref{eq:whittle_break_even} motivates the learned priority map. The deployed network combines five inputs, representing service pressure through measured violation feedback.

\subsection{From dual prices to violation feedback}
\label{subsec:time_varying_duals}

The fixed-dual analysis explains why persistent service deficits should influence scheduling priorities. A projected dual adaptation would increase a throughput multiplier when measured throughput falls below its target, and a regularity multiplier when measured TSLS exceeds its target:
\begin{align}
\mu^B_n(t+1)&=\left[\mu^B_n(t)+\eta_t(B_n-\widehat R^B_n(t))\right]^+,
\label{eq:dual_update_tpt}\\
\mu^L_n(t+1)&=\left[\mu^L_n(t)+\eta_t(\widehat R^L_n(t)-L_n)\right]^+.
\label{eq:dual_update_tsls}
\end{align}
Here $\widehat R^B_n,\widehat R^L_n$ are windowed service estimates and $\eta_t$ is a step size. For any fixed multiplier values, Assumption~\ref{assump:fixed_dual_indexability} provides the corresponding per-user indices. Ranking them and selecting the top $R$ gives a feasible high-action allocation.

Windex represents service pressure through measured violation feedback. Its five-input network combines queue, channel, and TSLS state with throughput and regularity violation statistics. The feedback changes over measurement windows, while the network ranks UEs every TTI. This retains the per-user computation and shared-resource selection suggested by the relaxation. The scalar theorem and fixed-dual assumption provide the analytical foundation, while the evaluation tests the resulting contextual scheduler.

%% file: 05-design_v2.tex
\section{Windex and CWindex Algorithm Design}
\label{sec:design}

Section~\ref{sec:solution} motivates a scheduler that computes a break-even resource price for each user and schedules the users with the largest prices. Windex learns this price with one compact neural index network per service class. CWindex optionally infers the service class from backlog history when labels are unavailable.

\subsection{Learning Windex networks}
Exact indices from~\eqref{eq:per_user_bellman} are impractical in the full RAN state. For service class $k$, Windex therefore learns
\begin{equation}
    f_{\theta,k}(x) \approx W_k(s),
\end{equation}
where $s=(q,c,\tau,d^B,d^L)$ contains backlog, CQI, TSLS, and normalized throughput/regularity deficits. The network represents deficits through empirical violation features, giving the training input
\begin{equation}
    x_t=(q_t,c_t,\tau_t,v_t^{\mathrm{tp}},v_t^{\mathrm{tsls}}),
    \label{eq:windex_features}
\end{equation}
where the violation-context features $v_t^{\mathrm{tp}},v_t^{\mathrm{tsls}}$ are varied over $[0,1]^2$ during offline training so that the network learns to adapt priorities to different service pressures.

Training follows NeurWIN~\cite{neurwin}. For a sampled resource price $\lambda$, the network induces the stochastic threshold policy
\begin{equation}
    \Pr(a_t=1\mid x_t,\lambda)=\sigma_m(f_{\theta,k}(x_t)-\lambda),
    \label{eq:windex_policy}
\end{equation}
Here $\sigma_m(z)=1/(1+e^{-mz})$. The price $\lambda$ and context $(v^{\mathrm{tp}},v^{\mathrm{tsls}})$ are sampled once per episode.  

The training reward is
\begin{equation}
R_t = w_r \bar r^B(s_t,a_t) + w_B v^{\mathrm{tp}} + w_L v^{\mathrm{tsls}} - \lambda a_t,
\label{eq:training_reward}
\end{equation}
Here $\bar r^B$ is normalized throughput, and $w_r,w_B,w_L$ are weights that depend on the service class. 

\begin{algorithm}[t]
\caption{WINDEX index-network training (one service class)}
\label{algo:Windex}
\begin{algorithmic}[1]
\STATE \textbf{Input:} network $f_\theta$, sigmoid slope $m$, batch size $R$, episode length $T$, weights $(w_r,w_B,w_L)$, price/context distributions $\Lambda,\mathcal V$.
\FOR{training batch $b$}
    \STATE Set $\bar G_b\leftarrow0$.
    \FOR{episode $e=1,\ldots,R$}
        \STATE Sample $\lambda\sim\Lambda$ and $(v^{\mathrm{tp}},v^{\mathrm{tsls}})\sim\mathcal V$.
        \STATE Initialize queue, CQI, TSLS; set $G_e\leftarrow0$, $h_e\leftarrow0$.
        \FOR{TTI $t=1,\ldots,T$}
            \STATE Form $x_t=(q_t,c_t,\tau_t,v^{\mathrm{tp}},v^{\mathrm{tsls}})$.
            \STATE Draw $a_t\sim\pi_\theta(\cdot\mid x_t,\lambda)$ using~\eqref{eq:windex_policy}.
            \STATE Apply $a_t$, observe throughput and TSLS, and update simulator state.
            \STATE $G_e\leftarrow G_e+\gamma^{t-1}R_t$ using~\eqref{eq:training_reward}.
            \STATE $h_e\leftarrow h_e+\nabla_\theta\log\pi_\theta(a_t\mid x_t,\lambda)$.
        \ENDFOR
        \STATE $\bar G_b\leftarrow\bar G_b+G_e/R$.
    \ENDFOR
    \STATE $\theta\leftarrow\theta+\alpha_b\sum_{e=1}^{R}(G_e-\bar G_b)h_e$.
\ENDFOR
\STATE \textbf{Output:} learned index network $f_\theta$.
\end{algorithmic}
\end{algorithm}

\begin{figure}[b]
\centering
\includegraphics[width=\columnwidth]{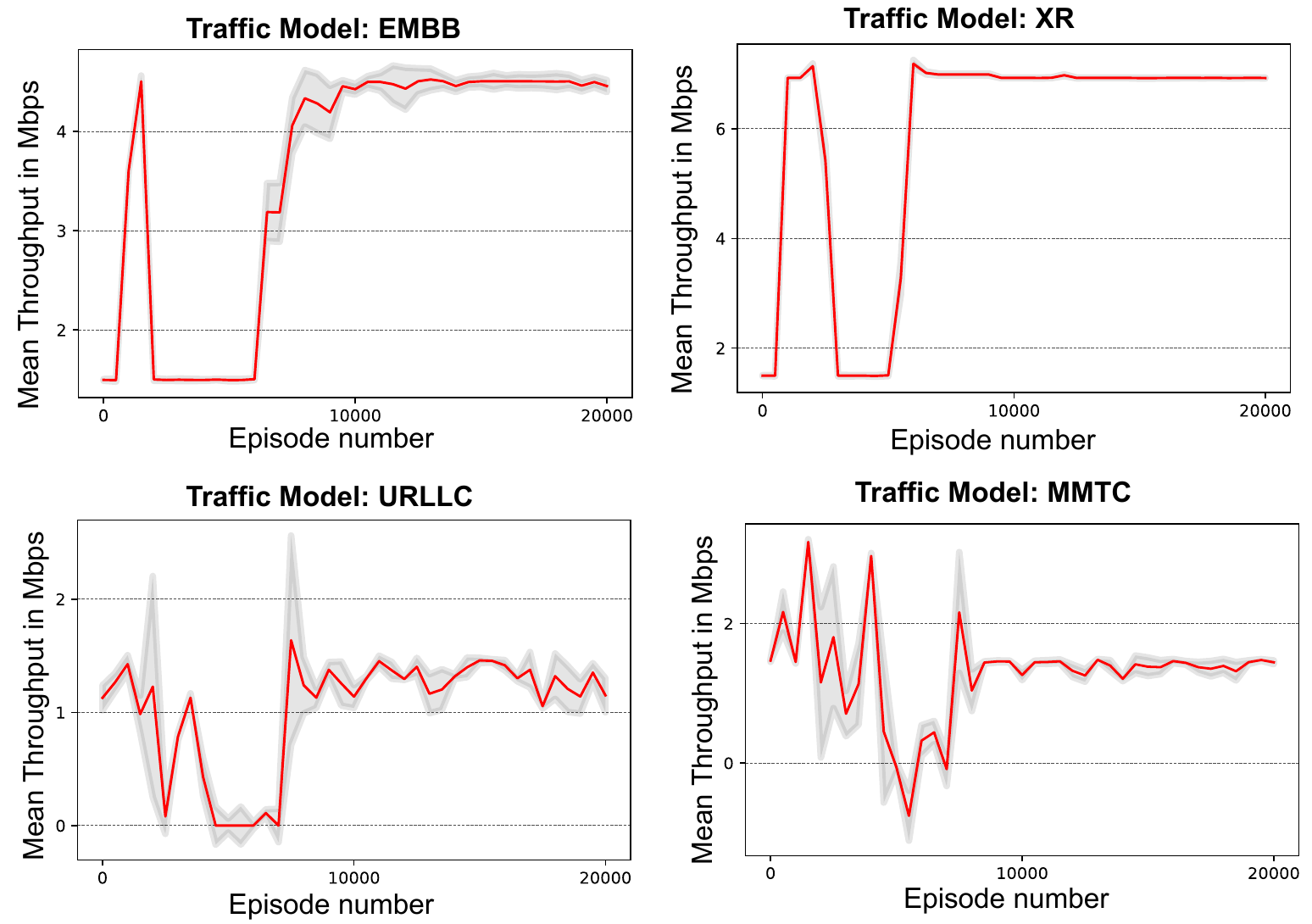}
\caption{Training curves for Windex service classes.}
\Description{Training-return curves for the service-class-specific neural index networks, illustrating the evolution of the training objective over episodes.}
\label{fig:windex_training}
\end{figure}

\begin{table}[t]
\small
\centering
\caption{Hyperparameters for Windex network training.}
\label{tab:Hyperparameters}
\begin{tabular}{@{}lp{0.69\columnwidth}@{}}
\toprule
Hyperparameter & Details\\
\midrule
Learning rate & 0.1 eMBB, 0.75 URLLC, 0.1 XR, 0.25 mMTC\\
Hidden layers & $(32,8)$\\
Batch size & 20 episodes\\
Optimizer & Adam\\
$(w_r,w_B,w_L)$ & \raggedright $(0.2,0.6,0.2)$ eMBB/XR; $(0.2,0.2,0.6)$ URLLC/mMTC\tabularnewline
\bottomrule
\end{tabular}
\end{table}

We train one network per service class for 20,000 episodes of 5,000 TTIs. eMBB and XR use sustained arrivals, while URLLC and mMTC use bursty arrivals. Figure~\ref{fig:windex_training} shows representative training curves averaged over 10 seeds.

\subsection{Online scheduling and feedback}

At each TTI, Windex evaluates~\eqref{eq:windex_features} and selects the top $R$ users to enforce the high-action budget. After each measurement window, observed throughput and TSLS violations update the context inputs and influence subsequent priorities. The measured feedback adapts scheduling priorities to realized service outcomes.

\begin{algorithm}[t]
\caption{WINDEX scheduler for $N$ UEs}
\label{algo:WindexScheduler}
\begin{algorithmic}[1]
\STATE \textbf{Input:} class networks $f_{\theta,k}$, budget $R$, targets $B_n,L_n$.
\STATE Initialize violation percentages $v^{\mathrm{tp}}_{n,0},v^{\mathrm{tsls}}_{n,0}$.
\FOR{each TTI $t$}
    \FOR{each user $n$}
        \STATE Use the latest windowed violation percentages $(v^{\mathrm{tp}}_{n,t}, v^{\mathrm{tsls}}_{n,t})$.
        \STATE Observe $q_n,c_n,\tau_n$ and form $x_n=(q_n,c_n,\tau_n,v^{\mathrm{tp}}_{n,t},v^{\mathrm{tsls}}_{n,t})$.
        \STATE Compute $w_n=f_{\theta,k(n)}(x_n)$.
    \ENDFOR
    \STATE Schedule the $R$ users with largest $w_n$ and map the high/low decisions to RBGs.
    \STATE Observe throughput and TSLS outcomes.
    \STATE Update queues and accumulate service statistics for all users.
\ENDFOR
\end{algorithmic}
\end{algorithm}

\subsection{CWindex: Inferring service-class}
\label{sec:cwindex_design}
Windex uses reliable QoS flow identifiers (QFIs), core-network policy, or application labels directly to select a service profile. CWindex supplies context when these labels do not reach the real-time scheduler, are stale or too coarse, or no longer match the observed traffic. Packet timing leaves patterns in queue-backlog histories, providing an in-band signal from which to infer service class. This preserves per-TTI scheduling with a slower context update.

\textbf{Signal used for classification.}
For user $n$, CWindex observes
\begin{equation}
    z_n(t)=\big(q_n(t-L+1),\ldots,q_n(t)\big),
\end{equation}
a window of $L$ backlog samples. Queue dynamics reveal service class because sustained eMBB/XR flows, bursty URLLC traffic, and low-rate mMTC traffic produce distinct temporal patterns.

\textbf{S4 classifier.}
CWindex uses a structured state-space model (S4)~\cite{gu2021efficiently}. An S4 layer starts from $x'(t)=\mathbf A x(t)+\mathbf B u(t)$ and $y(t)=\mathbf Cx(t)+\mathbf D u(t)$, then discretizes to $x_t=\bar{\mathbf A}x_{t-1}+\bar{\mathbf B}u_t$ and $y_t=\bar{\mathbf C}x_t+\mathbf D u_t$. This captures long-range structure with lower memory cost than full attention. After normalization, a two-layer S4 encoder, pooling, and softmax head produce
\begin{equation}
    \widehat k_n(t)=\argmax_k p_\psi(k\mid z_n(t)).
\end{equation}
The priority used by the scheduler is
\begin{equation}
    w_n(t)=f_{\theta,\widehat k_n(t)}(q_n,c_n,\tau_n,v_n^{\mathrm{tp}},v_n^{\mathrm{tsls}}).
\end{equation}
CWindex runs less frequently than Windex because service class changes more slowly than queue or channel state.

\textbf{Training and accuracy.}
We train CWindex from simulated backlog traces. Each simulation runs for $10^6$ TTIs. We divide traces into windows of $L=1024$ samples and sample 50,000 windows per class for eMBB, URLLC, and XR. The classifier uses cross-entropy loss, learning rate $10^{-3}$, weight decay $10^{-2}$, batch size 256, 10 epochs, two S4 layers, model dimension $d_{model}=8$, dropout 0.1, and pre-normalization. An mMTC extension would use its own training windows and classifier output.

%% file: 06-system.tex
\section{System Design}
\label{sec:system}

Windex runs as a real-time $\mu$App in a RAN Intelligent Controller (RIC). Figure~\ref{fig:solutionarchitecture} shows the RAN--RIC feedback path.

\subsection{Inference components}
\begin{figure}[!htbp]
\centering
\includegraphics[width=\columnwidth]{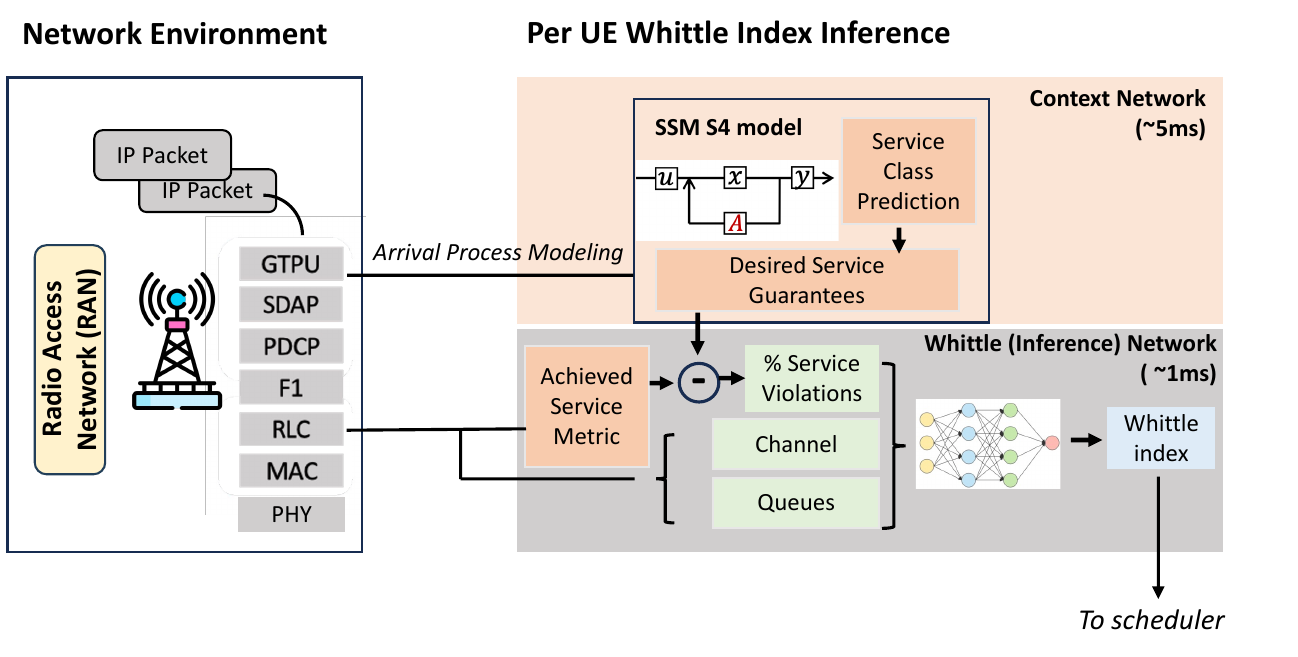}
\caption{System architecture. The RAN reports per-UE queue, channel, and service statistics to the real-time RIC. Windex computes per-UE indices every TTI. CWindex infers service class over longer backlog windows.}
\Description{The base station reports per-user queue, channel, and service measurements to EdgeRIC. Windex returns priorities every TTI. An optional slower context network infers service classes from backlog histories.}
\label{fig:solutionarchitecture}
\end{figure}
\textbf{Whittle inference network.}
The Whittle network runs every TTI. For each active UE, it receives $(q,c,\tau,v^{\mathrm{tp}},v^{\mathrm{tsls}})$ and returns a scalar index. Because inference is per-user, the computation parallelizes directly across UEs. The RIC sends the resulting priority vector or selected high-action set back to the RAN.

\textbf{Context network.}
CWindex selects a service-specific Windex network from backlog history at a slower cadence. When reliable labels are supplied by the core network, traffic-flow rules, or application signaling, Windex uses them directly.

\subsection{Implementation}
\textbf{Radio access network.}
We use the open-source srsRAN stack~\cite{srs} to instantiate a programmable 5G base station and software UEs. Each UE runs a traffic generator corresponding to a service profile. The MAC scheduler accepts Windex priorities and maps high/low actions to RBG allocations.

\textbf{Real-time RIC.}
We use EdgeRIC~\cite{edgeric} for low-latency control over the real-time E2 (RT-E2) interface. Each TTI, the RAN reports queue backlog, CQI, and service statistics, and Windex returns a policy message. CWindex reuses the reports for context inference.

\textbf{Traffic profiles.}
The evaluation uses four representative classes: URLLC bursty traffic with an average burst rate of 2~Mbps, mMTC bursty traffic with an average burst rate of 3.5~Mbps, eMBB continuous traffic at 5.8~Mbps, and XR continuous traffic at 6.2~Mbps. Each class has its own training weights and service targets.

\textbf{Deployment setup.}
The testbed operates over a 5~MHz channel in the 2.5~GHz EBS band using FDD and 1~ms TTIs. The reported system experiments use the real-time RIC and MAC control loop with channels driven by over-the-air traces. Replaying the same measured traces gives all schedulers identical channel conditions.

%% file: 07-evaluations.tex
\section{Evaluation}
\label{sec:evaluation}

We evaluate Windex as a real-time MAC scheduling policy using trace-driven emulation and simulation. The trace-driven experiments replay channel-quality traces collected from over-the-air measurements, including mobile, vehicular, drone, and indoor-robot settings. Replaying the same traces for all policies isolates the scheduler effect from channel randomness. Table~\ref{tab:scenarios} summarizes the traffic mixes and resource budgets.

\begin{table}[t]
\small
\centering
\caption{Evaluation scenarios: 1--3 use measured traces and physical RBs, while 4--7 list allocated simulator RBGs.}
\label{tab:scenarios}
\begin{tabular}{@{}lrrrrrrr@{}}
\toprule
Scenario & 1 & 2 & 3 & 4 & 5 & 6 & 7\\
\midrule
eMBB UEs & 1 & 2 & 2 & 2 & 5 & 10 & 20\\
XR UEs & 1 & 1 & 0 & 2 & 5 & 10 & 20\\
URLLC UEs & 1 & 0 & 1 & 2 & 5 & 10 & 20\\
Resource units & 25 & 25 & 25 & 17 & 47 & 94 & 188\\
\bottomrule
\end{tabular}
\end{table}

\subsection{Baselines and metrics}
The original scenario comparisons use a common resource pool and the same high/low action interface for all schedulers. Each baseline computes a scalar priority per UE, selects the top users under the high-action budget, and uses the same RBG mapping as Windex. The slicing experiment evaluates a hierarchical allocation.

\textbf{Max CQI} prioritizes the UE with the highest instantaneous channel quality. \textbf{Proportional fairness} ranks users by their current CQI divided by its exponentially weighted moving average. \textbf{Max-weight} uses $w_i[t]=CQI_i[t]B_i[t]$, where $B_i[t]$ is the downlink backlog. \textbf{Round robin} cycles through UEs independent of channel and queue state. \textbf{NVS} follows the virtualized-resource scheduling rule of~\cite{kokku2011nvs}. \textbf{QoSAware} assigns priorities using 3GPP service classes~\cite{5Gstandard,zaki2011multi}. \textbf{CWindex} uses the same scheduler as Windex but selects service classes from backlog history.

We report two violation metrics. A throughput violation occurs when a user's measured throughput over the evaluation window falls below its target. TSLS excess measures how far the user's time-averaged TSLS exceeds the service-class threshold. These metrics directly correspond to the constraints in~\eqref{eq:cmdp_tpt}--\eqref{eq:cmdp_tsls}.

\begin{figure}[t]
\centering
\includegraphics[width=\columnwidth]{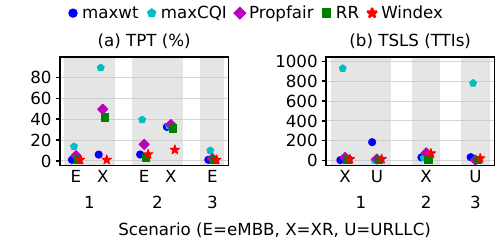}
\caption{Trace-driven throughput and service-regularity comparisons for Scenarios~1--3. Windex combines queue, channel, TSLS, and violation feedback to balance these objectives. Plotted values use a display floor of 1.}
\Description{Throughput and service-regularity violation comparisons for three three-user service mixes, with Windex and four traditional scheduling baselines.}
\label{fig:windex_evaluation}
\end{figure}

\begin{figure}[t]
\centering
\includegraphics[width=\columnwidth]{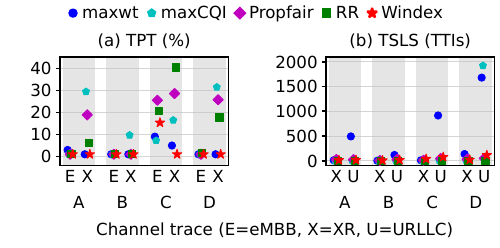}
\caption{Scenario~1 evaluated under four channel traces \mbox{(A--D)}. Plotted values use a display floor of 1.}
\Description{Throughput and service-regularity comparisons under four measured channel traces labeled A through D, with the same three-user service mix for each scheduler.}
\label{fig:windex_channel_evaluation}
\end{figure}

\subsection{Per-user service guarantees}
Scenarios~1--3 examine how a scheduler balances throughput and regularity when users compete for a shared resource pool. Figure~\ref{fig:windex_evaluation} evaluates Scenarios~1--3 on over-the-air channel traces. The 5~MHz setup has 25 physical RBs, grouped into 13 RBGs. The emulator maps high/low/no-service actions to 9/2/0 RBGs. Scenario~1 combines XR, eMBB, and URLLC. Scenarios~2 and~3 pair two eMBB users with XR and URLLC, respectively. We omit mMTC from these stressed scenarios because its low-rate, loose-regularity requirements are easy to satisfy under the tested load.

Windex reduces XR throughput shortfall by 6--89 percentage points relative to the baselines in Scenario~1 and by about 20 points relative to the best traditional scheduler in Scenario~2. The TSLS comparison reflects a tradeoff: for example, round robin achieves lower TSLS in Scenario~2, while Windex better meets the throughput target. Combining channel, backlog, and regularity feedback allows Windex to balance these service objectives.
\subsection{Robustness to channel dynamics}
A policy trained in simulation must continue to prioritize service needs when the channel follows a different trajectory. Keeping the traffic mix fixed lets us examine this sensitivity separately from changes in the UE population. Figure~\ref{fig:windex_channel_evaluation} evaluates the Scenario~1 traffic mix on several measured CQI traces. These traces include different mobility and blockage patterns, and all policies see the same replayed trace. Windex remains robust even though the networks are trained in simulation rather than on each trace. For example, on Trace~C, Windex reduces XR throughput violations by roughly 5--40 percentage points relative to traditional schedulers. These results support joint use of channel, queue, and service feedback.

\subsection{CWindex end-to-end scheduling}
\begin{figure}[t]
\centering
\includegraphics[width=\columnwidth]{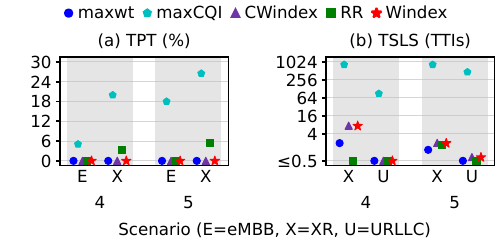}
\caption{End-to-end performance of CWindex on Scenarios~4--5. CWindex uses predicted service classes from backlog history, while Windex uses configured service labels.}
\Description{Throughput and service-regularity comparisons for 6 and 15 users, showing CWindex with inferred service labels alongside Windex and traditional schedulers.}
\label{fig:cwindex_eval}
\end{figure}
The practical test for CWindex is whether inferred service classes preserve scheduling performance. Classification accuracy summarizes label predictions, while end-to-end measurements reveal how those decisions affect service under competition for resources. Figure~\ref{fig:cwindex_eval} therefore compares configured and inferred context within the complete scheduling loop. Scenarios~4 and~5 contain 6 and 15 UEs, respectively, with two or five users in each of eMBB, XR, and URLLC. CWindex selects the class-specific Windex network from 1,024-sample backlog histories, while Windex uses configured service labels. Both use the same common-pool index scheduler, so their comparison captures the service impact of inferred context.

CWindex closely tracks Windex in throughput violations and remains well below Max-CQI in both throughput and TSLS violations. These results show that inferred context retains the main service benefits of Windex at the two tested populations.

\subsection{Scaling with the number of users}
\begin{figure}[t]
\centering
\includegraphics[width=\columnwidth]{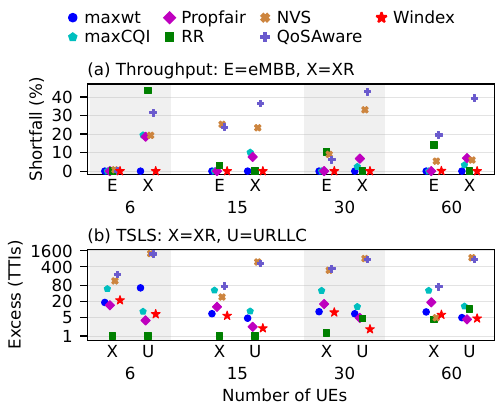}
\caption{Scaling from 6 to 60 UEs (Scenarios 4--7) with increasing resources. Metrics are computed after class averaging. Plotted TSLS values use a display floor of 1.}
\Description{Throughput shortfall for eMBB and XR and TSLS excess for XR and URLLC across 6, 15, 30, and 60 UEs, comparing Windex with six baseline schedulers.}
\label{fig:windex_scalability}
\end{figure}

Scenarios~4--7 increase the population from 6 to 60 UEs while scaling resources to maintain a constrained operating point. A synthetic random-walk channel makes the comparison repeatable. Figure~\ref{fig:windex_scalability} reports class-aggregate throughput shortfall, $100[B_k-\bar r_k]^+/B_k$, and class-mean TSLS excess, $[\bar\tau_k-L_k]^+$, where bars denote averages over users of class $k$ and the evaluation window. Here $B_k$ is the throughput target, set to 70\% of the class reference rate. At 60 UEs, Windex has 0.14\% eMBB shortfall and zero XR shortfall, with XR/URLLC TSLS excess of 6.35/4.60 TTIs. The per-user computation and common-pool selection remain unchanged as the population grows. This experiment verifies Windex's viability in a scaled-up scenario with appropriately increased provisioning.

\subsection{Slicing versus resource pooling}
\label{subsec:slicing}

Slicing assigns resources to service classes, but a slice budget alone cannot control which UE receives service at each TTI. This distinction matters because meeting average throughput targets does not bound the gaps between services. A hierarchical approach such as Zipper~\cite{balasingam2024application} allocates resources to slices and relies on an internal MAC scheduler to distribute them among UEs. Windex instead pools resources and incorporates individual throughput and TSLS feedback directly into its priorities.

We evaluate this Zipper-style architecture using max-weight within each slice, with 60 UEs (twenty per class), a 198-RBG cap, and 5,000-TTI episodes. For each partition in Table~\ref{tab:slicing}, an offline sweep selects feasible high/low counts per slice to minimize throughput-shortfall percentage plus mean TSLS excess. Max-weight ranks UEs every TTI by backlog times CQI. However, the selected allocation within a slice can leave capacity unused. These partitions evaluate slice allocation with a generic scheduler operating at every TTI.

\begin{table}[t]
\small
\centering
\caption{Slice budgets in RBGs for 60 UEs for a total of 198 RBGs. Windex uses a common pool of RBGs without slicing.}
\label{tab:slicing}
\begin{tabular}{@{}lrrr@{}}
\toprule
Configuration & eMBB & XR & URLLC\\
\midrule
1 & 66 & 66 & 66\\
2 & 88 & 66 & 44\\
3 & 66 & 110 & 22\\
4 & 44 & 66 & 88\\
\bottomrule
\end{tabular}
\end{table}

\begin{figure}[t]
\centering
\includegraphics[width=\columnwidth]{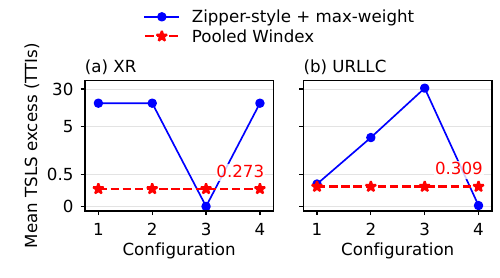}
\caption{Slicing versus pooling at 60 UEs. Each partition leaves substantial TSLS excess in XR or URLLC under the Zipper-style baseline. Windex's result is independent of the partition. The TSLS axis is linear below 0.5, then logarithmic.}
\Description{Mean positive per-user TSLS excess for XR and URLLC across four slice budgets. The hierarchical baseline has high XR excess in Configurations 1, 2, and 4 and high URLLC excess in Configuration 3. Windex keeps both near 0.3 TTIs.}
\label{fig:slicing_pooling}
\end{figure}

Both methods meet the eMBB and XR average throughput targets in these runs. Figure~\ref{fig:slicing_pooling} shows the mean positive per-user TSLS excess, $|U_k|^{-1}\sum_{n\in U_k}[\bar\tau_n-L_k]^+$. Every tested partition leaves substantial excess in at least one delay-sensitive class, whereas Windex keeps XR and URLLC excess at 0.27 and 0.31 TTIs. Favoring XR or URLLC can improve that class alone, but does not provide regular service to both. The result demonstrates the benefit of combining resource pooling with TSLS-aware per-user scheduling.

\subsection{Real-time feasibility}
\begin{figure}[t]
\centering
\includegraphics[width=\columnwidth]{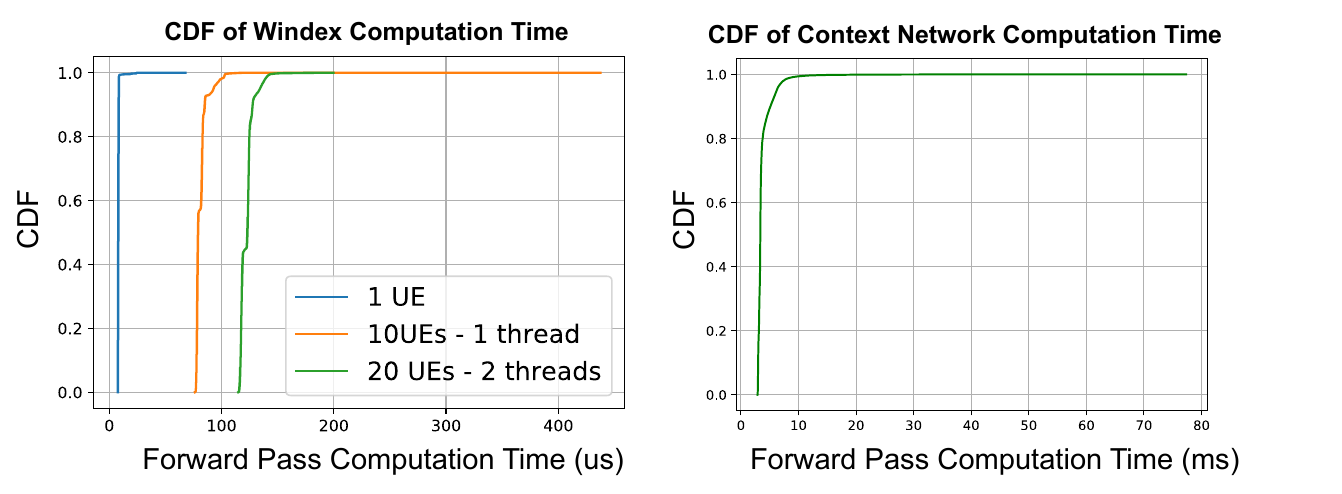}
\caption{Windex and context-network computation times.}
\Description{Measured computation-time distributions for neural index inference and the optional context classifier.}
\label{fig:windex_computation}
\end{figure}

\begin{figure}[!tbp]
\centering
\includegraphics[width=\columnwidth]{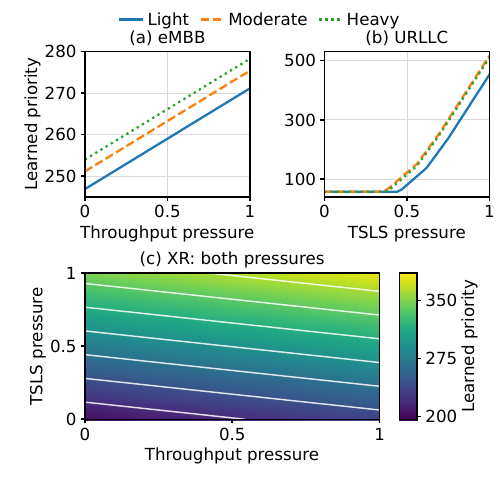}
\caption{Learned priorities versus violation pressure. (a,b) Normalized backlog/TSLS pairs are (0.1, 0.02), (0.5, 0.10), and (0.9, 0.30), with CQI 0.867 and unswept pressure 0.5. (c) XR at backlog 150,000 bytes, CQI 8, and TSLS 50 TTIs.}
\Description{The upper panels show eMBB and URLLC directional pressure sweeps at three fixed states. The lower heatmap jointly varies throughput pressure horizontally and TSLS pressure vertically for XR at a fixed physical state. Priority increases in either pressure direction.}
\label{fig:priority_pressure}
\end{figure}

A real-time scheduler must complete inference and policy transmission within the 1~ms TTI. Windex requires 10--20~$\mu$s per UE on a general-purpose CPU. Figure~\ref{fig:windex_computation} shows that Windex computes priorities for 20 UEs in under 150~$\mu$s using two threads. Accurate traffic identification with CWindex requires about one second of backlog history (1024 TTIs), followed by the network computation in Figure~\ref{fig:windex_computation}. These slower context updates are optional when service labels are available. Index computation scales linearly with the number of users, requiring one network evaluation per UE. Their independence allows parallel execution across CPU threads.

\subsection{Learned-priority diagnostics}

Windex uses violation feedback to direct resources toward unmet requirements. This requires a user's priority to rise with the corresponding service pressure at a fixed backlog, channel quality, and physical TSLS. We test this response by varying the violation-context inputs in~\eqref{eq:windex_features}. Figure~\ref{fig:priority_pressure}(a,b) shows that eMBB priority increases with throughput pressure and URLLC priority increases with TSLS pressure at the stated physical states.

XR combines throughput and regularity requirements, so we vary both pressures over $[0,1]^2$. Figure~\ref{fig:priority_pressure}(c) shows that increasing either pressure raises priority while the other input remains fixed. The scaling simulator computes XR priority as 1.3 times the eMBB network's output. At the displayed state, ReLU activation status remains unchanged across the pressure range, giving an affine response with positive throughput/TSLS pressure slopes of 31.56/146.01. A broader $51\times51$ pressure-grid check across 90 physical states finds zero decreasing steps among 459,000 adjacent grid-point comparisons. Thus, increasing either pressure raises priority at each evaluated physical state, providing the intended response to persistent service shortfalls.

%% file: 08-discussions.tex
\section{Conclusion}
\label{sec:conclusion}

Heterogeneous applications require sufficient throughput and regular service opportunities, since meeting an average rate alone can still leave users waiting too long between transmissions. We address these requirements in Windex through per-user priorities that combine queue, channel, and service-violation feedback, allocating resources from a common pool at each TTI. We use the Lagrangian decomposition and indexability results to motivate per-user priorities that we compute quickly and efficiently within a TTI. We further show that we can infer service classes from time series of queue backlogs using a state-space model in CWindex.

\balance

Through trace-driven emulation and simulations, we show that Windex balances throughput and service regularity across traffic mixtures, including 60 UEs with proportionally increased resources. We find that the tested slice partitions meet average eMBB and XR throughput targets but leave substantial XR or URLLC TSLS excess, while pooled Windex keeps both near 0.3 TTIs. Using two CPU threads, we compute indices for 20 UEs in under 150~$\mu$s. We thus translate application requirements into millisecond scheduling through per-user feedback, resource pooling, and scalar priorities.

In future work, we plan longer over-the-air trials with broader traffic classes to evaluate adaptation and context inference. We will also study admission control and target adaptation under overload.